\documentclass[conference,10pt]{IEEEtran}
\usepackage{cite}
\usepackage[cmex10]{amsmath}%
\usepackage{amssymb}%
\usepackage{amsthm}%
\usepackage{amsfonts}
\usepackage{algorithmic}
\usepackage[pdftex]{graphicx}
\usepackage{textcomp}
\usepackage[T1]{fontenc}%
\usepackage{bm}
\usepackage{array}%
\usepackage{mdwmath}
\usepackage{mdwtab}
\usepackage[caption=false]{subfig}
\usepackage{url}
\usepackage{xcolor}
\usepackage[left=1.62cm,right=1.62cm,top=1.8cm]{geometry}

\providecommand{\U}[1]{\protect\rule{.1in}{.1in}}
\newtheorem{theorem}{Theorem}
\newtheorem*{theorem*}{Theorem}

\newtheorem*{argument*}{Argument}
 
\usepackage{bbm}
\usepackage{xpatch}
\makeatletter
\xpatchcmd{\@thm}{\thm@headpunct{.}}{\thm@headpunct{}}{}{}
\makeatother

\DeclareMathOperator{\trace}{Tr}

\DeclareMathOperator{\diag}{diag}

\newcommand{\ket}[1]{\left| #1 \right>} 
\newcommand{\bra}[1]{\left< #1 \right|} 
\mathchardef\mhyphen="2D 

\IEEEoverridecommandlockouts
    
\begin{document}

\newcommand{\papertitle}{No-Go Theorem for Ancilla-Assisted Gaussian Enhancement in Passive-Unitary Estimation}
\title{\papertitle{}}

\author{\IEEEauthorblockN{Zihao Gong
        and Saikat Guha}
\IEEEauthorblockA{Electrical and Computer Engineering Department, University of Maryland, College Park, MD}
}

  \maketitle
\begin{abstract}
We study the maximum quantum Fisher information (QFI) for estimating a single parameter embedded in a generic multimode lossless passive Gaussian unitary using general Gaussian probes under a signal-energy constraint. Unlike previous work, which imposed a total energy constraint on the full probe, we constrain only the transmitted signal modes while allowing an arbitrary number of locally retained ancilla modes with arbitrarily large energy. We prove that this additional freedom does not increase the maximum achievable QFI;  the optimum remains identical to that attainable without extra ancilla energy. The same conclusion also extends to the sequential  setting under a total energy constraint. We also characterize the family of optimal probe states and show that entanglement is not necessary to attain the optimum in the lossless setting. This extends the result of Matsubara \textit{et al.} to the physically motivated signal-energy-constrained scenario and establishes a no-go theorem for ancilla-assisted Gaussian enhancement in noiseless passive-unitary estimation.
\end{abstract}

\section{Introduction} \label{sec:introduction}
\par

The fundamental precision limits of parameter estimation are studied in quantum estimation theory~\cite{helstrom1969quantum}. In particular, for a given probe state, the quantum Fisher information (QFI) determines the ultimate precision of estimating a single parameter achievable by any measurement and unbiased estimator. Quantum metrology seeks probe states and measurements that use quantum resources to enhance estimation precision~\cite{giovannetti2006,giovannetti2011advances}. Optimizing the QFI over the available probe states is therefore a central task in quantum parameter estimation.

Gaussian states form an important class of continuous-variable quantum states. They admit a simple description in terms of  first and  second moments, and are experimentally relevant in quantum optics~\cite{Braunstein2005,Weedbrook2012}. In addition, Gaussian resources such as squeezing can enhance metrological performance; a landmark example is the squeezed-state interferometric sensitivity improvement introduced by Caves in 1981~\cite{Caves1981}. The QFI of Gaussian states has been extensively studied and can be expressed directly in terms of the first and second moments of the state in phase space~\cite{monras2013phase,serafini2023quantum}. These properties make Gaussian states a natural class of  probe states for continuous-variable quantum metrology.

We consider the problem of maximizing the QFI for estimating a single parameter embedded in a generic multimode lossless passive Gaussian unitary using the most general Gaussian-state probe  under a signal-energy constraint. Such unitaries describe photon-number-preserving linear-optical circuits, including multimode interferometers built from beam splitters and phase shifters~\cite{Clements2016}. The single parameter of interest may be a function of these phase shifts.
 The most general pure Gaussian state can be generated from a tensor product of single-mode squeezed coherent states by passive Gaussian unitaries~\cite{serafini2023quantum}. Since the QFI is convex in the probe state, it suffices to optimize over pure Gaussian probes.

The closest prior work is due to Matsubara \textit{et al.}, who studied the same estimation problem under a total mean photon-number constraint on the probe state and identified the maximum achievable QFI~\cite{matsubara2019optimal}. In this work, we instead consider a physically motivated setting that we constrain only the energy of the transmitted signal modes, while allowing the locally retained ancilla modes to carry arbitrary energy. Although entanglement is known to enhance a variety of other quantum tasks~\cite{Saikat2020ISIT,gong22losssensing,Gong2025}, we show that this additional freedom does not increase the maximum achievable QFI.

The rest of this paper is organized as follows. In Sec.~\ref{sec:Problem_setup}, we formulate the ancilla-assisted Gaussian sensing problem under a signal-energy constraint and introduce the corresponding QFI expression. In Sec.~\ref{sec:Main_result}, we present the main theorem and characterize the family of optimal probe states. In Sec.~\ref{sec:proof}, we provide the detailed proofs of the upper bounds used in the main result.

\section{Problem setup} \label{sec:Problem_setup}
\begin{figure}[ht] 
    \centering
        \includegraphics[width=0.9\linewidth]{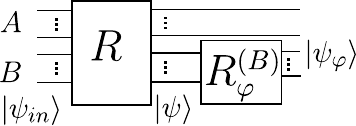}
  \caption{Setup for ancilla-assisted single-parameter estimation using Gaussian probes. The probe consists of $m_A$ ancilla modes and $m_B$ signal modes. The initial state $\ket{\psi_{in}}$ is taken to be a tensor product of single-mode squeezed coherent states. By applying a passive linear interferometer $R$, we obtain  the general pure Gaussian probe state $\ket{\psi}$. The $m_B$ signal modes are transmitted through the passive lossless Gaussian unitary $R_{\varphi} = I^{(A)}\otimes R_{\varphi}^{(B)}$, which encodes the unknown parameter $\varphi$, while the $m_A$ ancilla modes remain unchanged. The final output state is denoted by $\ket{\psi_{\varphi}}$.}
\end{figure}

\subsection{Probe state}
We consider a bipartite sensing problem using a Gaussian probe, where subsystem  $A$  consists of  $m_A$  ancilla modes and subsystem  $B$  consists of  $m_B$  signal modes. The signal subsystem  $B$  passes through a passive Gaussian unitary that encodes the unknown parameter  $\varphi$ , while the ancilla subsystem  $A$  is retained and may be entangled with the signal at the input.

To generate the general pure Gaussian probe state, we begin with a tensor product of single-mode squeezed coherent states $\ket{\psi_{in}}$  and then apply a passive Gaussian unitary  $R$.
 In the $\hat{q}\hat{q}\hat{p}\hat{p}$ form, its mean vector $ d_0 $ and covariance matrix $\Gamma_0 $ are given by
\begin{align}
   d_0 &=  \begin{bmatrix}
        d^{(A)}_{0}, \
        d^{(B)}_0,  \
        \Tilde{d}^{(A)}_0, \
        \Tilde{d}^{(B)}_0 
    \end{bmatrix}^T, \ \   \Gamma_0  = \frac{1}{2} e^{2r}  \\
    r &= \diag\left( r_A,r_B,-r_A,-r_B \right), \label{eq:r}
\end{align}
where  
 $ r_{A} = \diag\left(r_{A}^{(1)},r_{A}^{(2)},\dots,r_{A}^{(m_A)}\right)$, $r_{B} = \diag\left(r_{B}^{(1)},r_{B}^{(2)},\dots,r_{B}^{(m_B)}\right) $, and the superscript $T$ denotes the transpose of a matrix or a vector.
Without loss of generality, we assume $ r_{A}^{(1)}\ge r_{A}^{(2)} \ge\dots \ge r_{A}^{(m_A)}\ge 0 $, and $ r_{B}^{(1)}\ge r_{B}^{(2)} \ge\dots \ge r_{B}^{(m_B)} \ge 0$.
The passive Gaussian unitary $R$ has symplectic matrix representation:
\begin{align}
     R &= W^{\dagger}\begin{bmatrix}
        U &0\\
        0 &U^*
    \end{bmatrix}W, \ \ 
    W  = \frac{1}{\sqrt{2}} \begin{bmatrix}
        I &iI\\
        I  &-iI
    \end{bmatrix}, \ \
\end{align}
where $R$ is an orthogonal matrix, $I$ denotes the identity matrix, and $U$ is a unitary matrix decomposed as
\begin{align}
        U = \begin{bmatrix}
        U^{(A)} &U^{(AB)} \\
        U^{(BA)} &U^{(B)} 
    \end{bmatrix}. \label{eq:U}
\end{align}
From the unitary condition $U U^{\dagger}  =  I$, we have
\begin{align}
    U^{(BA)}U^{(BA)\dagger} + U^{(B)}U^{(B)\dagger}  = I^{(B)}, \label{eq:unitary condition}
\end{align}
where $I^{(B)}$ is $m_B \times m_B $ identity matrix.
After applying $R$, the resulting pure Gaussian probe state $\ket{\psi}$ has mean vector $d $ and covariance matrix  $\Gamma $:
\begin{align}
    d &=   \begin{bmatrix}
        d^{(A)} & d^{(B)} &\Tilde{d}^{(A)} & \Tilde{d}^{(B)} 
    \end{bmatrix}^T = R d_0 , \label{eq: M of rhop} \\
   \Gamma &=  R \Gamma_0 R^{\dagger} \\
   &= \frac{1}{2} W^{\dagger}\begin{bmatrix}
       U\cosh 2rU^{\dagger} &U \sinh 2r U^T \\
       U^* \sinh 2r U^{\dagger} & U^* \cosh 2r U^T
   \end{bmatrix}W. \label{eq: V of rhop}
\end{align}

\subsection{Energy constraint} 
Unlike the total-energy constraint used in prior work~\cite{matsubara2019optimal}, we constrain only the mean photon number in the transmitted signal subsystem $B$.
The mean photon number of the signal modes of the probe state $\ket{\psi}$ is upper bounded by $\bar{n}$, i.e.,
\begin{align}
     \frac{1}{2}\trace\left( \trace_A[\Gamma] - \frac{1}{2}I^{(B)}\right) + \frac{1}{2}d_{B}^{T}d_{B} \le \Bar{n} ,
\end{align}
 where $\trace_A$ denotes the partial trace over subsystem $A$. 
Write the energy in displacement as $\bar{n}_d^{(B)} = \frac{1}{2}d_{B}^{T}d_{B}$, where $d_B = [d^{(B)}, \tilde{d}^{(B)}]^T $ denotes the mean vector of subsystem $B$. The remaining  energy from  squeezing is upper bounded by
\begin{align}
     \frac{1}{2}\trace\left( \trace_A[\Gamma] - \frac{1}{2}I^{(B)}\right) 
   = \trace\left( M   \right) 
   \le  \bar{n} - \bar{n}_d^{(B)}, \label{eq:energy constraint}
\end{align}
where 
\begin{align}
  \!  M\! = \! U^{(B)} \sinh^2(r_B) U^{(B)\dagger} \!+ U^{(BA)} \sinh^2(r_A) U^{(BA)\dagger} \label{eq:M}
\end{align}
 follows from straightforward  algebra.


\subsection{Quantum Fisher information}
The parameter $\varphi$ is embedded  in the passive Gaussian unitary represented by the orthogonal matrix $R_{\varphi}$:
\begin{align}
     R_{\varphi} = W^{\dagger}\begin{bmatrix}
        U_{\varphi} &0\\
        0 &U_{\varphi}^*
    \end{bmatrix}W, \ \
    U_{\varphi} = \begin{bmatrix}
       I^{(A)} &0   \\
           0   &U_{\varphi}^{(B)} 
    \end{bmatrix}, \label{eq:R varphi}
\end{align}
where  $ U_{\varphi}$ and  $ U_{\varphi}^{(B)}$ are unitary.
Because the input probe is pure and  $R_\varphi$  is unitary, the output state  $\ket{\psi_{\varphi}}$  remains a pure Gaussian state with mean vector and covariance matrix 
\begin{align}
    d_{\varphi} = R_{\varphi}d, \ \ \ \Gamma_{\varphi} = R_{\varphi}\Gamma R_{\varphi}^{\dagger}.
\end{align}
 Thus, the QFI with respect to $\varphi$ is \cite{monras2013phase,serafini2023quantum}
\begin{align}
    F_{\varphi} = \frac{1}{4}\trace\left[ \left(\Gamma_{\varphi}^{-1} \frac{\partial \Gamma_{\varphi} }{\partial\varphi}\right)^2 \right] + \frac{ \partial d^{T}_{\varphi} }{\partial \varphi} \Gamma_{\varphi}^{-1} \frac{ \partial d_{\varphi} }{\partial \varphi}. \label{eq:QFI no loss}
\end{align}
Define the Hermitian generator associated
with the signal-mode unitary as
\begin{align}
    g_{\varphi} = i U_{\varphi}^{(B)\dagger}
\frac{\partial U_{\varphi}^{(B)}}{\partial \varphi}. \label{eq:g varphi}
\end{align}
Under a mean-photon-number constraint on the probe state,
Matsubara et al. showed that the optimal Gaussian QFI for this
passive-unitary estimation problem is
\begin{align}
    F_{\varphi,\mathrm{M}}^{\mathrm{opt}}
=
8\|g_{\varphi}^2\|\bar{n}(\bar{n}+1),
\end{align}
where $\|A\|$ is the spectral norm, i.e., the largest singular value of $A$.
In contrast, we impose the constraint only on the transmitted
signal energy and answer whether arbitrary-energy ancilla modes
can increase the optimal QFI.

\section{Main result} \label{sec:Main_result}
\subsection{Main theorem} \label{subsec:Main_theorem}
\begin{theorem}
Under a mean photon-number constraint $\bar{n}$ on the signal modes, the maximum QFI over all Gaussian probes for estimating a parameter embedded in a passive Gaussian unitary is
\begin{align}
    F_\varphi^{\mathrm{opt}} = 8 \|g_\varphi^2\| \bar{n}(\bar{n}+1).
\end{align}
  Hence, allowing an arbitrary number of ancilla modes and arbitrarily large ancilla energy does not improve the optimal QFI.
\end{theorem}
\begin{proof}
By the convexity of the QFI, it suffices to optimize over pure probe states.
Using \eqref{eq:R varphi}, we define the Hermitian matrix
\begin{align}
    \tilde{G}_{\varphi}  \!= \! i R_{\varphi}^\dagger \frac{\partial R_{\varphi}}{\partial \varphi}\!=\! W^{\dagger}\begin{bmatrix}
    \tilde{g}_{\varphi} & 0 \\
    0  & -\tilde{g}_{\varphi}^{*}
\end{bmatrix} W , \    \tilde{g}_{\varphi} = \begin{bmatrix}
    0 & 0 \\
    0  & g_{\varphi}
\end{bmatrix}, \label{eq:tilda g varphi}
\end{align} 
where $g_{\varphi}$ is in \eqref{eq:g varphi}.
Substituting \eqref{eq:tilda g varphi} into the QFI formula in \eqref{eq:QFI no loss}, we decompose the QFI as
\begin{align}
    F_{\varphi} &=  F_{\varphi}^{(1)}+F_{\varphi}^{(2)},
    \end{align}
 \vspace{0.07in}   where $F_{\varphi}^{(1)}$ is the covariance-matrix contribution and $F_{\varphi}^{(2)}$ is the displacement contribution: \begin{align}
    F_{\varphi}^{(1)}& =    \frac{1}{2}\trace\left[  \tilde{G}_{\varphi}  \Gamma^{-1} \tilde{G}_{\varphi}\Gamma   - \tilde{G}_{\varphi}^2\right]\label{eq: QFI 1} \\
     F_{\varphi}^{(2)}  &= d^{T} \tilde{G}_{\varphi} \Gamma^{-1}  \tilde{G}_{\varphi}d. \label{eq: QFI 2}
\end{align}
The displacement $d$ and covariance matrix $ \Gamma$  are defined in \eqref{eq: M of rhop} and \eqref{eq: V of rhop}, respectively, and
\begin{align}
\Gamma^{-1} = 2 W^{\dagger} \begin{bmatrix}
       U \cosh{2r} U^{\dagger} &   -U \sinh{2r} U^{T} \\
       -U^* \sinh{2r} U^{\dagger} &   U^* \cosh{2r} U^{T}
    \end{bmatrix} W. \label{eq:Gamma^-1}
\end{align}
We first upper bound $F_{\varphi}^{(1)}$. 
Substituting the expressions for $\Gamma$ and $\Gamma^{-1}$ into \eqref{eq: QFI 1} yields
\begin{align}
    F_\varphi^{(1)}
    = \trace\left(\tilde{g}_\varphi X \tilde{g}_\varphi X
    + \tilde{g}_\varphi Y \tilde{g}_\varphi^T Y^{\dagger}\right)
    - \trace(\tilde{g}_\varphi^2), \label{eq:F1 1}
\end{align}
where
\begin{align}
    X = U \cosh2r U^\dagger, \qquad
    Y = U \sinh2r U^T, \label{eq:X Y}
\end{align}
Applying the inequalities $\trace(PQPQ)\le \trace(P^2Q^2)$ and $\trace(PQP^TQ^T)\le \trace(P^2Q^2)$ for Hermitian $P$ and $Q$, \eqref{eq:F1 1} yields
\begin{align}
    F_\varphi^{(1)}
    \le  2\,\trace\!\left[\tilde{g}_\varphi^2\left(X^2-I\right)\right]. \label{eq:F1 u2}
\end{align}
Expanding the right-hand side and using $\trace(PQ)\le \|P\|\trace(Q)$  for Hermitian and positive semidefinite $P$ and $Q$, we obtain 
\begin{align}
    F_\varphi^{(1)}
    \le 4 \|g^2_{\varphi} \| \trace \left(   M^2  + M\right), \label{eq:F1 u1}
\end{align}
where $M$ is in \eqref{eq:M}.
Using $ \trace(M^2)\le \trace(M)^2 $ for positive semidefinite $M$, together with the signal-energy constraint \eqref{eq:energy constraint}, \eqref{eq:F1 u1} yields
\begin{align}
    F_\varphi^{(1)}
    \le
    8\|g_\varphi^2\|
    \left[
        \left(\bar{n}-\bar{n}_d^{(B)}\right)^2
        +
        \left(\bar{n}-\bar{n}_d^{(B)}\right)
    \right]. \label{eq:F 1 upper bound}
\end{align}
The detailed derivation on upper bound of $  F_\varphi^{(1)} $ is given in Sec.~\ref{subsec:QFI 1}.

For the displacement contribution $F_\varphi^{(2)}$, using the Rayleigh quotient \cite{higham2008functions} together with standard norm inequalities, we obtain
\begin{align}
    F_\varphi^{(2)}
    \le
2\|g_\varphi^2\|\left(\|X_{22}\|+\|Y_{22}\|\right)\bar{n}_d^{(B)},
\end{align}
where 
\begin{align}
   \!\! X_{22} \! &= \!2 U^{(BA)} \cosh{2 r_A}U^{(BA)\dagger} \! + \! 2 U^{(B)} \cosh{2 r_B}U^{(B)\dagger} \label{eq:X22} \\
    Y_{22}\! & = \! -2U^{(BA)}\! \sinh\!{2 r_A}U^{(BA)T} \! \! - \! 2 U^{(B)} \!\sinh\!{2 r_B}U^{(B)T} \!\!. \label{eq:Y22}
\end{align}
We then bound $\|X_{22}\|$ using $\|P\|\le \trace(P)$ for Hermitian and positive semidefinite $P$, and bound $\|Y_{22}\|$ using Takagi's decomposition \cite{horn2012matrix} and the Cauchy--Schwarz inequality. This yields
\begin{align}
    F_\varphi^{(2)}
    \le
    16\|g_\varphi^2\|
    \left(\bar{n}-\bar{n}_d^{(B)}+\frac{1}{2}\right)\bar{n}_d^{(B)}. \label{eq:F 2 upper bound}
\end{align}
The detailed derivation on upper bound of $  F_\varphi^{(2)} $ is given in Sec.~\ref{subsec:QFI 2}.
 Combining \eqref{eq:F 1 upper bound} and \eqref{eq:F 2 upper bound}, we have
\begin{align}
    F_{\varphi}&\le 8\| g_{\varphi}^2 \| \left(  \left(\Bar{n} -\Bar{n}_{d}^{(B)}\right)^2  +  \left(\Bar{n} -\Bar{n}_{d}^{(B)}\right ) \right. \nonumber\\
    &\left.\qquad\qquad\qquad\qquad+  2\left(\Bar{n} -\Bar{n}_{d}^{(B)} +\frac{1}{2} \right) \Bar{n}_{d}^{(B)}  \right) \\
    & = 8\| g_{\varphi}^2 \|   \left( \bar{n}^2 +  \bar{n} -\bar{n}_d^{(B)2}\right) \label{eq: F upper 0} \\
    & \le 8\| g_{\varphi}^2 \| \Bar{n}(\Bar{n} +1)  .\label{eq: F upper 3}
\end{align}


Finally, we observe that the resulting upper bound on the QFI, $8 \|g_{\varphi}^2 \| \Bar{n}(\Bar{n}+1)$, coincides with the QFI achievable when no additional energy is allocated to the ancilla modes~\cite{matsubara2019optimal}. As shown in Sec.~\ref{subsec:optimal_probe}, this upper bound is achievable. Hence, the optimal QFI under the signal-energy constraint is $8\|g_\varphi^2\|\bar{n}(\bar{n}+1)$, and unconstrained ancilla modes provide no improvement.

\end{proof}

\subsection{Optimal probe states attaining the QFI bound}\label{subsec:optimal_probe}
We now characterize probe states that attain the upper bound in \eqref{eq: F upper 3}. To saturate this bound, the signal displacement energy is set to be zero, i.e.,
\begin{align}
    \bar{n}_d^{(B)}=0.
\end{align}
Hence the optimal probe has zero displacement, so that $d=0$ and $F_\varphi^{(2)}=0$. The optimization therefore reduces to maximizing the covariance contribution $F_\varphi^{(1)}$.
First, the equality condition for $\trace(M^2)\le\trace(M)^2 $, used in obtaining \eqref{eq:F 1 upper bound}, requires $M$ to be rank one. 
Without loss of generality, we place the active mode in the first position, so that 
\begin{align}
    r_A = \diag(r_A^{(1)},0,\dots,0), \  \ \
     r_B = \diag(r_B^{(1)},0,\dots,0).
\end{align} 
Second, the equality in \eqref{eq:F1 u1} requires that the support of $M$ be contained in the eigenspace of $g_\varphi^2$ associated with its largest eigenvalue. 
We now use this alignment condition to determine $U$ in \eqref{eq:U}. 
First, we take 
\begin{align}
    m_A = m_B = \dim (g_{\varphi}) = m.
\end{align}
Applying the cosine-sine decomposition \cite{paige1981towards} to the blocks of $U$, we obtain unitary matrices $V_B$, $P_B$, $V_A$, and $P_A$, together with diagonal matrices
\begin{align}
    C=\diag(c_1,\dots,c_m),\qquad
    S=\diag(s_1,\dots,s_m),
\end{align}
such that
\begin{align} 
    U^{(A)} &= P_A C V_B^\dagger,
    \qquad
    U^{(AB)} = -P_A S V_A^\dagger \\
    U^{(B)}& = P_B S V_B^\dagger,
    \qquad
    U^{(BA)} = P_B C V_A^\dagger,
\end{align}
with $ C^2+S^2=I$.
Let $U_g$ be the unitary matrix such that $ U_g^{\dagger} g_\varphi^2U_g  $ is diagonal. To ensure that the support of $M$ is contained in the eigenspace of $g_\varphi^2$ associated with its largest eigenvalue, we choose
\begin{align}
    P_B=U_g,\quad V_B=V_A=I,\quad C=cI,\quad S=sI,
\end{align}
with $c^2+s^2=1$. Then, by the unitary condition, we take $P_A = U_g D$, where $D$ is a diagonal phase matrix. Thus, 
\begin{align}
    U = \begin{bmatrix}
        cU_g D & -sU_g D \\
        sU_g & cU_g
    \end{bmatrix}.
\end{align}
For a probe that saturates the QFI upper bound, the squeezing-energy constraint in  \eqref{eq:energy constraint} and \eqref{eq:M} reduces to
\begin{align}
    c^2 \sinh^2 \left(r_A^{(1)}\right) + s^2 \sinh^2 \left( r_B^{(1)}\right) = \bar{n}. \label{eq:QFI condition}
\end{align}
Once the above two equality conditions hold, it is straightforward to verify that the remaining inequalities in deriving \eqref{eq:F1 u2} are also saturated.
Therefore, we obtain a family of probes that saturate the QFI upper bound in \eqref{eq: F upper 3}.

Here, we present two representative examples. The first example is 
\begin{align}
   \!\!\! c \! = \! 0, \ s \!= \!1, \  r_A^{(1)}  \! =\! 0, \  \sinh^2 \left( r_B^{(1)}\right) \! = \! \bar{n}, \ U \!=\! \begin{bmatrix}
       0 & 0 \\
        0 & U_g
    \end{bmatrix},
\end{align}
 which reduces to Matsubara's optimal probe state. The second example is  
 \begin{align}
     c = \frac{1}{\sqrt{2}}, \  s = \frac{1}{\sqrt{2}}, \ &\sinh^2 \left(r_A^{(1)}\right) = \sinh^2 \left( r_B^{(1)}\right) = \bar{n}, \\
     U &=  \frac{1}{\sqrt{2}}\begin{bmatrix}
         U_g   & - U_g   \\
         U_g &  U_g
    \end{bmatrix},
 \end{align}
 which corresponds to a generalized two mode squeezed vacuum state. These examples further show that, under the signal-energy constraint \eqref{eq:energy constraint}, entanglement is not necessary to attain the optimal QFI in the lossless setting. 
 By contrast, entanglement may become beneficial in the presence of loss and noise.

\subsection{Sequential extension}

We also consider the $L$ sequential setting studied by Matsubara \textit{et al.}, in which the same parameter-dependent channel is interrogated $L$ times sequentially with known passive Gaussian unitary controls inserted between each channel use.  Using the same argument as in the proof of the theorem, one finds that the norm of the generator matrix is upper bounded by $L\|g_\varphi\|$. Therefore, the corresponding QFI is upper bounded by
\begin{align}
    F_\varphi^{(L)} \le 8L^2\|g_\varphi^2\|\, \bar{n}(\bar{n} +1),
\end{align}
which coincides with the bound obtained in~\cite{matsubara2019optimal}. Hence, entanglement does not improve the optimal QFI in the sequential setting either.

\section{Proofs} \label{sec:proof}
\subsection{Bound on $F_{\varphi}^{(1)}$}\label{subsec:QFI 1}
In this subsection, we derive the upper bound on  $ F_{\varphi}^{(1)} $ in \eqref{eq: QFI 1}.
Substituting the expressions for $\tilde{G}_{\varphi}$ in \eqref{eq:tilda g varphi}, $\Gamma$ in \eqref{eq: V of rhop}, $X$ and $Y$ in \eqref{eq:X Y}, and $\Gamma^{-1}$ in \eqref{eq:Gamma^-1} into \eqref{eq: QFI 1}, we obtain
\begin{align}
    F_{\varphi}^{(1)}  &= \frac{1}{2}\trace  \left( \tilde{g}_{\varphi} X \tilde{g}_{\varphi} X + \tilde{g}_{\varphi} Y \tilde{g}_{\varphi}^{T} Y^{\dagger} + \tilde{g}_{\varphi}^{T}  Y^{\dagger}\tilde{g}_{\varphi} Y \right. \nonumber \\
    & \left. \qquad\qquad+ \tilde{g}_{\varphi}^T X^T \tilde{g}_{\varphi}^T X^T \right) - \frac{1}{2}\trace\left[\tilde{g}_{\varphi} ^2 + (\tilde{g}_{\varphi}^{T})^2 \right].
\end{align}
Using $\trace(A^T) = \trace(A)$, we obtain
\begin{align}
    F_{\varphi}^{(1)} & =  \trace  \left( \tilde{g}_{\varphi} X \tilde{g}_{\varphi} X  +\tilde{g}_{\varphi} Y \tilde{g}_{\varphi}^{T} Y^{\dagger}   \right) -  \trace\left[\tilde{g}_{\varphi} ^2     \right].\label{eq:QFI first} 
\end{align}
We first bound the first and the third terms in \eqref{eq:QFI first}. Using $\trace(PQPQ) \le \trace(P^2 Q^2)$ with $P = U^{\dagger} \tilde{g}_{\varphi} U  $ and $Q = \cosh 2r $, we have
\begin{align}
    &\quad\trace  \left( \tilde{g}_{\varphi}  U \cosh (2r)  U^{\dagger} \tilde{g}_{\varphi}  U \cosh (2r) U^{\dagger}      \right) -  \trace\left[\tilde{g}_{\varphi} ^2     \right] \nonumber\\
    &  \le \trace  \left( \tilde{g}_{\varphi}^2 \left( U \cosh^2 (2r)  U^{\dagger} - I \right)      \right)   \\
    & = \trace  \left( \tilde{g}_{\varphi}^2 \left( \left( U \cosh (2r)  U^{\dagger} \right)^2 - I \right)      \right). \label{eq: QFI1  11}
\end{align}
Next, we consider the second term in \eqref{eq:QFI first}, using $\trace(PQP^TQ^T)\le \trace(P^2Q^2)$ with $P = U^{\dagger} \tilde{g}_{\varphi} U  $ and $Q = \sinh (2r) $, we obtain
    \begin{align}
        &\quad\trace\left(\tilde{g}_{\varphi} U \sinh (2r)\, U^{T} \tilde{g}_{\varphi}^{T} U^{*} \sinh (2r) U^{\dagger} \right) \nonumber
        \\
        &\le  \trace\left(\tilde{g}_{\varphi}^2 \left( U \sinh (2r)  U^{\dagger} \right)^2 \right) \\
        &= \trace  \left( \tilde{g}_{\varphi}^2 \left( \left( U \cosh (2r)  U^{\dagger} \right)^2 - I \right)      \right). \label{eq: QFI1  12}
    \end{align}
Substituting \eqref{eq: QFI1  11} and \eqref{eq: QFI1  12} into \eqref{eq:QFI first}, we obtain
\begin{align}
    F_{\varphi}^{(1)}\le 2 \trace  \left( \tilde{g}_{\varphi}^2 \left( \left( U \cosh (2r)  U^{\dagger} \right)^2 - I \right)      \right). \label{eq:QFI1 13}
\end{align}
Substituting the expressions for $\tilde{g}_{\varphi}$ in \eqref{eq:tilda g varphi}, $U$ in \eqref{eq:U}, and $r$ in \eqref{eq:r} into the right hand side of \eqref{eq:QFI1  13}, we have
\begin{align}
  F_{\varphi}^{(1)}  &\le 2\trace \left( g^2_{\varphi}  \left(\left(
             U^{(B)} \cosh(2r_B) U^{(B)\dagger} \right.\right.\right. \nonumber\\
    & \qquad\left.\left.\left.        
            +
              U^{(BA)} \cosh(2r_A) U^{(BA)\dagger}
        \right)^2 - I^{(B)}\right)\right)\\
        & = 8 \trace \left( g^2_{\varphi}  \left(M^2 + M
        \right)\right). \label{eq:F1 u 3}
\end{align}
where $M$ is given in \eqref{eq:M}.
Since $ \sinh^2(r_A) $ and $\sinh^2(r_B) $ are positive semidefinite,  $ M $ and $ M^2 $ are Hermitian and positive semidefinite. 
Moreover, since $ g^2_{\varphi} $ is also Hermitian and positive semidefinite, using $\trace(PQ)\le \|P\|\trace(Q)$, \eqref{eq:F1 u 3} yields
\begin{align}
   F_{\varphi}^{(1)} \le & 8\|g^2_{\varphi} \| \trace \left(   M^2 + M
        \right). \label{eq:QFI1 12.5}
\end{align} 
The right hand side of \eqref{eq:QFI1 12.5} is further bounded using $\trace(M^2) \le \trace(M)^2$ for positive semidefinite $M$ and \eqref{eq:energy constraint},
\begin{align}
   F_{\varphi}^{(1)} \le & 8 \|g^2_{\varphi} \| \left[
        \left(\Bar{n} - \bar{n}_d^{(B)}\right)^2
        +
        \left(\Bar{n} - \bar{n}_d^{(B)}\right)
    \right].\label{eq:QFI1 14}
\end{align}
\subsection{Bound on $F_{\varphi}^{(2)}$}\label{subsec:QFI 2}
We now bound $ F^{(2)}_{\varphi}$ in \eqref{eq: QFI 2}. 
A direct block multiplication using $\tilde{G}_{\varphi}$ in \eqref{eq:tilda g varphi},  and $\Gamma^{-1}$ in \eqref{eq:Gamma^-1} gives
\begin{align}
    F_{\varphi}^{(2)} 
    & = d_B^{T} W_B^{\dagger}  \begin{bmatrix}
       g_{\varphi} X_{22}g_{\varphi}    & -g_{\varphi} Y_{22}g_{\varphi}^* \\
       - g_{\varphi}^* Y^*_{22} g_{\varphi}  &g_{\varphi}^* X_{22}^* g_{\varphi}^*
    \end{bmatrix} W_B d_B,
\end{align}
where  $W_B = \frac{1}{\sqrt{2}} \begin{bmatrix}
        I^{(B)} &iI^{(B)}\\
        I^{(B)}  &-iI^{(B)}
    \end{bmatrix}$, $X_{22}$ is in \eqref{eq:X22}, and $Y_{22}$ is in \eqref{eq:Y22}.
Furthermore, 
    \begin{align}
        \!\!\!\!F_{\varphi}^{(2)} \!&\le \left\| \begin{bmatrix}
         g_{\varphi}  &0 \\
        0  &-g_{\varphi}^*  
    \end{bmatrix} \begin{bmatrix}
        X_{22}    &  Y_{22} \\
        Y^*_{22}   & X_{22}^*
    \end{bmatrix}\begin{bmatrix}
         g_{\varphi}  &0 \\
        0  &-g_{\varphi}^*  
    \end{bmatrix}  \right\|   d_B^{T} d_B  \label{eq:QFI2 11} \\
    &  \le 2 \| g_{\varphi}^2\| \left\| \begin{bmatrix}
        X_{22}    &  Y_{22} \\
        Y^*_{22}   & X_{22}^*
    \end{bmatrix}\right\|    \Bar{n}_{d}^{(B)} \label{eq:QFI2 12} \\
     & \le 2 \| g_{\varphi}^2\| \left(\left\| \begin{bmatrix}
        X_{22}    & 0 \\
        0   & X_{22}^*
    \end{bmatrix}\right\| + \left\|\begin{bmatrix}
       0    &  Y_{22} \\
        Y^*_{22}   & 0
    \end{bmatrix}\right\| \right)   \Bar{n}_{d}^{(B)} \label{eq:QFI2 13} \\
    & =  2\| g_{\varphi}^2\| \left(\left\|  X_{22}   \right\| + \left\|Y_{22} \right\| \right)   \Bar{n}_{d}^{(B)}. \label{eq:QFI2 14}
    \end{align}
Here, \eqref{eq:QFI2 11} follows from Rayleigh quotient, \eqref{eq:QFI2 12} from the submultiplicative of the spectral norm together with $\Bar{n}_{d}^{(B)} = \frac{1}{2}d_B^{T} d_B$, and \eqref{eq:QFI2 13} from the triangle inequality of matrix norm. Finally, \eqref{eq:QFI2 14} follows from
\begin{align}
    \left\|\begin{bmatrix}
       0    &  Y_{22} \\
        Y^*_{22}   & 0
    \end{bmatrix}\right\|^2 &=  \left\|\begin{bmatrix}
       0    &  Y_{22} \\
        Y^*_{22}   & 0
    \end{bmatrix} \begin{bmatrix}
       0    &  Y_{22} \\
        Y^*_{22}   & 0
    \end{bmatrix}^{\dagger} \right\| \nonumber \\
    &= \left\| Y_{22} Y_{22} ^{\dagger} \right\| = \left\| Y_{22}  \right\|^2
\end{align}
We now upper bound  $\|X_{22}\|$,
\begin{align}
   \!\! \|X_{22}\| 
    &\! = \! 4  \left\|  M \right\| \! + \! 2  \!\le\! 4  \trace\! \left(M\right) \!+\! 2 \!=\! 4\! \left(\Bar{n} -\Bar{n}_{d}^{(B)} \right) \!+\!2.\label{eq:QFI2 23}
\end{align}
The first equality follows from  $\cosh 2x = 2 \sinh^2 x +1$ and \eqref{eq:unitary condition}. The inequality comes from $\|M\| \le \trace (M)$ for positive semidefinite $M$, and the last equality follows from \eqref{eq:energy constraint}.

Next, we upper bound $\|Y_{22}\|$.
Define the zero-mean input annihilation operator vectors $\bm{\hat{a}}_{A} = [\hat{a}_A^{(1)} - \langle \hat{a}_A^{(1)} \rangle,\dots,\hat{a}_A^{(m_A)}- \langle \hat{a}_A^{(m_A)} \rangle ]^T$ and $\bm{\hat{a}}_{B} = [\hat{a}_B^{(1)} - \langle \hat{a}_B^{(1)} \rangle,\dots,\hat{a}_B^{(m_B)} - \langle \hat{a}_B^{(m_B)} \rangle]^T$. Since the state $\ket{\psi_{in}}$ is the tensor product of single mode squeezed vacuum states, and writing $\langle \cdot\rangle_{\ket{\psi_{in}}\bra{\psi_{in}}}$ as $\langle\cdot\rangle$, we have
\begin{align}
    \sinh 2 r_A  & = -2 \left\langle \bm{\hat{a}}_{A} \bm{\hat{a}}_{A}^T  \right\rangle , \ \ 
  \sinh 2 r_B   = -2 \left\langle \bm{\hat{a}}_{B} \bm{\hat{a}}_{B}^T \right \rangle  \label{eq:QFI L2}
\end{align}
The annihilation operator vector for the probe state is 
\begin{align}
    \bm{\hat{b}}_{B} = U^{(BA)} \bm{\hat{a}}_{A} + U^{(B )} \bm{\hat{a}}_{B}.
\end{align}
Using  $\langle \hat a_A^{(k)} \hat a_B^{(\ell)} \rangle = 0$  for all  $k,\ell$, $Y_{22}$ satisfies
\begin{align}
 \langle \bm{\hat{b}}_{B} \bm{\hat{b}}_{B}^T \rangle & \! = \! \left\langle   U^{(BA)} \bm{\hat{a}}_{A} \bm{\hat{a}}_{A}^T U^{(BA)T} \!+ \!U^{(B )} \bm{\hat{a}}_{B} \bm{\hat{a}}_{B}^T U^{(B )T}  \right\rangle \label{eq:QFI L3} \\
     & = \frac{1}{4}  Y_{22} .
\end{align}
By Takagi's decomposition of the symmetric matrix $ \langle \bm{\hat{b}}_{B} \bm{\hat{b}}_{B}^T \rangle$, there exists a unitary matrix $U_C$ such that 
\begin{align}
    U_C \langle \bm{\hat{b}}_{B} \bm{\hat{b}}_{B}^T \rangle U_C^{T} =  \frac{1}{4} U_C   Y_{22} U_C^{T} = \langle \bm{\hat{\tilde{b}}}_{B} \bm{\hat{\tilde{b}}}_{B}^T \rangle
\end{align}
is diagonal, where $ \bm{\hat{\tilde{b}}}_{B} = U_C  \bm{\hat{ b }}_{B} $. Therefore, 
\begin{align}
    \|Y_{22}\| 
    &=  \sqrt{\left\|U_C Y_{22} U_C^{T} U_C^{*} Y_{22}^{\dagger} U_C^{\dagger}\right\|} \label{eq:QFI2 24} \\
    &= 4\sqrt{\left\|\langle \bm{\hat{\tilde{b}}}_{B} \bm{\hat{\tilde{b}}}_{B}^T \rangle \langle \bm{\hat{\tilde{b}}}_{B} \bm{\hat{\tilde{b}}}_{B}^T \rangle^{\dagger} \right\|} \\
    &= 4 \left\| \langle \bm{\hat{\tilde{b}}}_{B} \bm{\hat{\tilde{b}}}_{B}^T \rangle \right\|, \label{eq:QFI2 2}
\end{align}
where \eqref{eq:QFI2 24} and \eqref{eq:QFI2 2} follow from $ \|P\| = \sqrt{\|PP^{\dagger}\|}$, and \label{eq:QFI2 25} uses unitary invariance of the norm and $ U_C^{T} U_C^{*} = I $.
Since $\langle \bm{\hat{\tilde{b}}}_{B} \bm{\hat{\tilde{b}}}_{B}^T \rangle $ is diagonal, we have
\begin{align}
     \left\| \langle \bm{\hat{\tilde{b}}}_{B} \bm{\hat{\tilde{b}}}_{B}^T \rangle \right\| & = \max_k \langle  \hat{\tilde{b}}^{(k)}_{B}  \hat{\tilde{b}}^{(k)}_{B} \rangle \label{eq:QFI2 31}\\
     &\le \max_k   \sqrt{\langle \hat{\tilde{b}}^{(k)\dagger}_{B}  \hat{\tilde{b}}^{(k)}_{B} \rangle \langle\hat{\tilde{b}}^{(k)}_{B}  \hat{\tilde{b}}^{(k)\dagger}_{B}\rangle} \label{eq:QFI2 32} \\
     &\le \sqrt{  \left(\Bar{n} -\Bar{n}_{d}^{(B)}\right)\left(\Bar{n} -\Bar{n}_{d}^{(B)}+1\right) } \label{eq:QFI2 33} \\
     &< \Bar{n} -\Bar{n}_{d}^{(B)}+\frac{1}{2}, \label{eq:QFI2 3}
\end{align}
where \eqref{eq:QFI2 32} follows from Cauchy-Schwarz inequality $ \langle l_2 | l_1 \rangle \le \sqrt{\langle l_1 | l_1 \rangle \langle l_2 | l_2 \rangle } $ with $\ket{l_1} = \hat{\tilde{b}}^{(k)}_{B} \ket{\psi_{in}} $ and $\ket{l_2} = \hat{\tilde{b}}^{(k)\dagger}_{B} \ket{\psi_{in}} $, \eqref{eq:QFI2 33} uses that $ \langle \hat{\tilde{b}}^{(k)\dagger}_{B}  \hat{\tilde{b}}^{(k)}_{B} \rangle $, the mean photon number in squeezing of a single mode of the probe state, is upper bounded by the total energy in squeezing, and \eqref{eq:QFI2 3} follows from $ \sqrt{x(x+1)} < \sqrt{x(x+1)+\frac{1}{4}} = x+\frac{1}{2}   $, for $x>0$. 

Substituting \eqref{eq:QFI2 23}, \eqref{eq:QFI2 2}, and \eqref{eq:QFI2 3} into \eqref{eq:QFI2 14},
we conclude that
\begin{align}
    F_{\varphi}^{(2)} &\le  16\|g_{\varphi}^2\| \left(\Bar{n} -\Bar{n}_{d}^{(B)} +\frac{1}{2}\right) \Bar{n}_{d}^{(B)}.
\end{align}

\section{Discussion} \label{sec:discussion}
We have shown that, under a signal-energy constraint, allowing arbitrarily many ancilla modes and arbitrarily large ancilla energy does not increase the maximum achievable QFI for estimating a single parameter embedded in a lossless passive Gaussian unitary. Hence, the ultimate precision in this setting is fully determined by the signal energy budget, and entanglement is not necessary to attain the optimum. The same no-go conclusion also extends to the $L$-use sequential setting under a total-energy constraint. This extends the total-energy-constrained result of Matsubara \textit{et al.} to the more physically motivated regime in which only the transmitted signal energy is constrained.

Ancilla assistance is expected to become useful in the presence of loss and noise. It would also be interesting to determine whether a similar no-go result holds in multiparameter estimation.

\newpage
\bibliographystyle{IEEEtran}
\bibliography{papers}

\end{document}